\documentclass[a4paper,11pt]{amsart}
\usepackage[pdfdisplaydoctitle=true,colorlinks=true,urlcolor=blue,citecolor=blue,linkcolor=blue,pdfstartview=FitH,pdfpagemode=None,bookmarksnumbered=true]{hyperref}

\newtheorem{theorem}{Theorem}
\newtheorem{corollary}[theorem]{Corollary}
\newtheorem{lemma}[theorem]{Lemma}
\newtheorem{proposition}[theorem]{Proposition}

\theoremstyle{definition}
\newtheorem{definition}[theorem]{Definition}

\theoremstyle{remark}
\newtheorem{remark}{Remark}

\newcommand{\Real}{\mathbb{R}}

\newcommand{\Circle}{\mathbb{S}^{1}}

\newcommand{\e}{\mathfrak{exp}}
\newcommand{\id}{\mathrm{id}}
\newcommand{\uo}{u_{0}}

\newcommand{\g}{\mathfrak{g}}

\newcommand{\DiffS}{\mathrm{Diff}^{\infty}(\mathbb{S}^{1})}
\newcommand{\Diff}[1]{\mathrm{Diff}^{#1}(\mathbb{S}^{1})}
\newcommand{\VectS}{\mathrm{Vect}^{\infty}(\mathbb{S}^{1})}
\newcommand{\VectR}{\mathrm{Vect}^{\infty}(\Real)}
\newcommand{\DiffR}{\mathrm{Diff}^{\infty}(\Real)}
\newcommand{\CS}{\mathrm{C}^{\infty}(\mathbb{S}^{1})}
\newcommand{\Vect}{\mathrm{Vect}}
\newcommand{\C}[1]{\mathrm{C}^{#1}(\mathbb{S}^{1})}

\newcommand{\norm}[1]{\left\Vert#1\right\Vert}

\newcommand{\set}[1]{\left\{#1\right\}}

\hypersetup{
    pdftitle={The Degasperis-Procesi equation as a non-metric Euler equation}, 
    pdfauthor={J. Escher and B. Kolev}, 
    pdfsubject={MSC 2000: 35Q53, 58D05}, 
    pdfkeywords={Euler equation, diffeomorphisms group of the circle, Degasperis-Procesi equation}, 
    }

\begin{document}

\title[non-metric Euler equation]{The Degasperis-Procesi equation as a non-metric Euler equation}%

\author[J. Escher]{Joachim Escher} %
\address{Institute for Applied Mathematics, University of Hannover, D-30167 Hannover, Germany} %
\email{escher@ifam.uni-hannover.de} %

\author[B. Kolev]{Boris Kolev} %
\address{CMI, 39 rue F. Joliot-Curie, 13453 Marseille cedex 13, France} %
\email{kolev@cmi.univ-mrs.fr}%




\begin{abstract}
In this paper we present a geometric interpretation of the periodic Degasperis-Procesi equation
as the geodesic flow of a right invariant symmetric linear connection on the diffeomorphism group of the circle. We also show that for any evolution in the family of $b$-equations there is neither gain nor loss of the spatial regularity of solutions. This in turn allows us to view the Degasperis-Procesi and the Camassa-Holm equation as an ODE on the Fr\'{e}chet space of all smooth
functions on the circle.
\end{abstract}

\maketitle


\section{Introduction}
\label{sec:intro}

Due to the highly involved structure of the full governing equations for the classical water wave problem, it is intriguing to approximate them by mathematically easier models. In the shallow-water medium-amplitude regime, introduced to capture stronger nonlinear effects than in the small-amplitude regime (which leads to the famous Korteweg-de Vries equation), the Camassa-Holm (CH) equation~\cite{CH93} and the Degasperis-Procesi (DP) equation \cite{DP99} attracted a lot of attention due to their integrable structure as infinite bi-Hamiltonian systems \cite{CGI06,DHH02} and to the fact that their solitary wave solutions are solitons.
Moreover it is known that this medium-amplitude regime allows for breaking waves, cf. \cite{CL09,Joh02, CE98b}

Both of these equations are members of the so-called family of `$b$-equations' \cite{DHH02,HW03}
\begin{equation}\label{eq:b-equation}
    m_{t} = - (m_{x}u + bmu_{x})
\end{equation}
where
\begin{equation*}
    m = u - u_{xx} .
\end{equation*}
For $b=2$ it is the Camassa-Holm equation (CH)
\begin{equation}\label{eq:CH}
    u_{t} - u_{txx} + 3uu_{x} - 2u_{x}u_{xx} - uu_{xxx} = 0
\end{equation}
and for $b=3$ it is the Degasperis-Procesi equation (DP)
\begin{equation}\label{eq:DP}
    u_{t} - u_{txx} + 4uu_{x} - 3u_{x}u_{xx} - uu_{xxx} = 0 .
\end{equation}

The hydrodynamic relevance of the $b$-equations is described e.g. in~\cite{Joh03,Iva07}. It is integrable only for $b = 2$ and $b = 3$ (see \cite{HW03,MN02,Iva05}). The periodic Camassa-Holm equation is known to correspond to a \emph{geodesic flow on the diffeomorphism group of the circle}, cf. \cite{Kou99}. Local existence of the geodesics and properties of the Riemannian exponential map were studied in \cite{CK02, CK03}. The main objective of the present paper is to extend this work to the whole family of b-equations and in particular for (DP).

\medskip

Analogous to the Camassa-Holm equation, see \cite{CE98, RB01}, the Cauchy problem for the $b$-equation is locally well-posed in the Sobolev space $H^s$ for any $s>3/2$ and the solution depends continuously on the initial data, cf. \cite{ELY06b, EY08}. The geometric approach of the present paper enables us to complete this picture by showing that the $b$-equation is well-posed in the smooth category $C^\infty(J,C^\infty(\Circle))$, where $J$ is an open interval containing $0$. The situation of analytic initial data and the question of analyticity of the corresponding solution will be addressed in a forthcoming paper.

We already mentioned several common properties of the (CH) and the (DP) equation. Besides the fact that the connection on $\DiffS$ in the geometric re-expression of the (DP) equation is non-metric, it is worth pointing out that the (DP) equation possesses bounded but discontinuous solitons, so-called shock peakons, cf. \cite{E07}. Indeed, it is possible to verify that
\begin{equation*}
    u_c(t,x)=\frac{\sinh\left(x-[x]-{1/2}\right)}{t\cosh\left({1/2}\right)+c\sinh\left({1/2}\right)}, \quad x\in\Real/\mathbb{Z},
\end{equation*}
is for any $c>0$ a weak solution to the (DP) equation. Note that the (CH) equation has no discontinuous weak solutions, due to the fact that the $H^1$-norm of any weak solution is preserved.

Another significant difference becomes clear by comparing the isospectral problem in the Lax pair for both equations. Letting again $m = u-u_{xx}$, we have the third order equation
\begin{equation*}
    \psi_x - \psi_{xxx} - \lambda m \psi = 0
\end{equation*}
as the isospectral problem for the (DP) equation (cf. \cite{DHH02}), while in the case of the (CH) equation, the second order equation
\begin{equation*}
    \psi_{xx} - \frac{1}{4}\psi - \lambda m \psi = 0
\end{equation*}
is the relevant spectral problem in the Lax pair, cf. \cite{CH93}.

\medskip

The paper is organized as follows. In Section~\ref{sec:non-metric_Euler_equation}, we introduce the notion of non-metric Euler equation. Some basic material about the diffeomorphism group $\DiffS$ is recalled in Section~\ref{sec:diffeomorphisms_group}. In Section~\ref{sec:b-equation}, we show that the $b$-equation may be recast as the geodesic flow of a symmetric linear connection on $\DiffS$ and that the corresponding second order vector field is smooth. In Section~\ref{sec:short_time_existence}, we show that there is neither gain nor loss of the spatial regularity of solutions and we establish short time existence and smooth dependence on the initial data for this geodesic flow in the smooth category. In Section~\ref{sec:exponential_map}, we show that the exponential map of this symmetric linear connection is a smooth local diffeomorphism. In the last Section~\ref{sec:non-periodic_case}, we discuss the non-periodic case.


\section{Non-metric Euler equations}
\label{sec:non-metric_Euler_equation}

The idea of studying geodesic flows in order to analyze the motion of hydrodynamical systems is attributed to Arnold \cite{Arn66} (despite the fact that a previous short paper from Moreau in 1959 already discussed the problem \cite{Mor59}). He pointed out that the Euler equations of the motion of a rigid body and the Euler equations of hydrodynamics (with fixed boundary) could both be recast as the geodesic equations of a one-sided invariant Riemannian metric on a Lie group.

This structure is the prototype for the mathematical treatment of physical systems which underlying configuration space can be identified with a Lie group $G$. The general Euler equation was derived initially for the \emph{Levi-Civita connection} of a one-sided invariant \emph{Riemannian metric} on a Lie group $G$ (see \cite{Arn66} or \cite{AK98}) but the theory is valid in the more general setting of a one-sided invariant \emph{linear connection} on $G$, as we shall explain now.

A \emph{linear connection} (\emph{covariant derivative}) on a manifold $M$ is a bilinear map
\begin{equation*}
    \Vect(M)\times\Vect(M) \to \Vect(M),
\end{equation*}
denoted by $(X,Y) \mapsto \nabla_{X}Y$, such that:
\begin{equation*}
    \nabla_{fX}Y = f\, \nabla_{X}Y,
\end{equation*}
and
\begin{equation*}
    \nabla_{X}(fY) = (X\cdot f)Y + f \, (\nabla_{X}Y),
\end{equation*}
for all $X,Y \in \Vect(M)$ and $f\in C^{\infty}(M)$. The linear connection $\nabla$ is \emph{symmetric} if moreover
\begin{equation*}
    \nabla_{X}Y - \nabla_{Y}X = [X,Y]
\end{equation*}
for all $X,Y \in \Vect(M)$, where $[X,Y]$ is the usual bracket of vector fields.

Given a local chart $(x^{j})$ (or more generally, any local frame), a linear connection $\nabla$ is completely defined by the \emph{Christoffel symbols}
\begin{equation*}
    \Gamma_{ij}^{k} = \left( \nabla_{\partial_{i}}(\partial_{j})\right)^{k}.
\end{equation*}

A covariant derivative $\nabla$ can be uniquely extended for smooth vector fields $X(t)$ defined only along a smooth path $x(t)$ on $M$ (see \cite{Lan99} for instance). This operator, denoted by ${D}/{Dt}$, is an extension of $\nabla$ in the sense that if the vector field $X(t)$ is the restriction to the curve $x(t)$ of a globally defined vector field $X$, that is $X(t) = X(x(t))$ then we have
\begin{equation*}
    \frac{DX}{Dt}(t) = \left(\nabla_{\dot{x}}\, X\right) (x(t)).
\end{equation*}
In a local chart $(x^{j})$, ${D}/{Dt}$ is defined by
\begin{equation*}
    \left(\frac{DX}{Dt}\right)^{k} = \dot{X}^{k}  + \Gamma_{ij}^{k} \dot{x}^{i}X^{j}.
\end{equation*}

\begin{definition}
A \emph{geodesic} for the linear connection $\nabla$ is a smooth curve $x(t)$ on $M$ such that
\begin{equation*}
    \frac{D\dot{x}}{Dt} = 0.
\end{equation*}
Given a local chart $(x^{j})$ of $M$, the smooth curve $x(t)$ is a geodesic iff it satisfies the following second order differential equation
\begin{equation*}
    \ddot{x}^{k}+\Gamma_{ij}^{k}\dot{x}^{i}\dot{x}^{j} = 0.
\end{equation*}
\end{definition}

\begin{remark}
If two linear connections differ by a skew-symmetric tensor, they have the same geodesics.
\end{remark}

Let $\varphi$ be a diffeomorphism of $M$. We say that the linear connection $\nabla$ is \emph{invariant} under $\varphi$ if
\begin{equation*}
    \varphi^{*} \left(\nabla_{X} Y\right) = \nabla_{\varphi^{*}X}\varphi^{*}Y
\end{equation*}
where $\varphi^{*}X$ is the vector field defined by
\begin{equation*}
    (\varphi^{*}X)(x) = T_{\varphi(x)}\varphi^{-1}X(\varphi(x)).
\end{equation*}

Suppose now that $M$ is a Lie group $G$. Then, we say that a linear connection on $G$ is \emph{right invariant} (resp. \emph{left invariant}) if it is invariant under each right translation $R_{g}$ (resp. left translation $L_{g}$) of $G$. On any Lie group $G$, there is a canonical connection which is both right and left invariant (we say it is \emph{bi-invariant}). It is defined by
\begin{equation*}
    \nabla^{0}_{\xi_{u}}\xi_{v} = \frac{1}{2}[\xi_{u},\xi_{v}],
\end{equation*}
where $\xi_{u}$ is the right invariant vector field on $G$ generated by the vector $u$ in the Lie algebra $\g$ of $G$.

\begin{remark}
This definition does not change if we use \emph{left invariant} vector fields rather than right invariant vector fields in the preceding formula.
\end{remark}

Given a right-invariant linear connection $\nabla$ on $G$, then
\begin{equation*}
    B(X,Y) = \nabla_{X}Y - \nabla^{0}_{X}Y
\end{equation*}
is a right-invariant tensor field on $G$ and is therefore completely defined by its value at the unit element $e$ of $G$, that is by a bilinear map $\g \times \g \to \g$. Conversely, each bilinear operator
\begin{equation*}
    B : \g \times \g \to \g
\end{equation*}
defines uniquely a \emph{right invariant linear connection} on the Lie group $G$, which is given by
\begin{equation}\label{eq:connection}
    \nabla_{\xi_{u}} \, \xi_{v} = \frac{1}{2} [\xi_{u},\xi_{v}] + B(\xi_{u},\xi_{v}),
\end{equation}
where we use the same notation for $B$ and the right invariant tensor field it generates on $G$.

Given a basis $(e_{k})_{1 \le k \le n}$ of the Lie algebra $\g$, let $(\xi_{k})$ be the global right invariant frame on $G$ generated by the basis $(e_{k})$ of $\g$ and let $(\omega^{k})$ be the dual \emph{co-frame} of $(\xi_{k})$ (notice that each $\omega^{k}$ is itself a right invariant 1-form on $G$).

Let $g(t)$ be a smooth path in $G$. We define its \emph{Eulerian velocity} by
\begin{equation*}
    u(t) = R_{g^{-1}(t)}\dot{g}(t)
\end{equation*}
which is a smooth path in the Lie algebra $\g$ of $G$. Let $u^{k}$ be the components of $u$ in the basis $(e_{k})$. Then
\begin{equation*}
    u^{k} = \omega^{k}_{e}(u) = \omega^{k}_{g}(R_{g}u) = \omega^{k}_{g}(\dot{g})
\end{equation*}
are the components of the vector $\dot{g}$ in the frame $(\xi_{k})$. The linear operator ${D}/{Dt}$ defined along $g(t)$ is then given by
\begin{equation*}
    \left(\frac{DX}{Dt}\right)^{k} = \dot{X}^{k}  + \left( \frac{1}{2}c_{ij}^{k} + b_{ij}^{k} \right)u^{i}X^{j}
\end{equation*}
where $c_{ij}^{k}$ are the structure constants of the Lie algebra $\g$ and $b_{ij}^{k}$ are the components of the tensor $B$. In particular, a smooth curve $g(t)$ in $G$ is a geodesic if and only if
\begin{equation*}
    \dot{u}^{k} + b_{ij}^{k} u^{i}u^{j} = 0
\end{equation*}
because $c_{ji}^{k} = - c_{ij}^{k}$ for all $i,j,k$.

\begin{proposition}
A smooth curve $g(t)$ on a Lie group $G$ is a \emph{geodesic} for a \emph{right invariant linear connection} $\nabla$ defined by~\eqref{eq:connection} iff its Eulerian velocity $u = R_{g^{-1}}\dot{g}$ satisfies the first order equation
\begin{equation}\label{eq:euler}
    u_{t} = - B(u,u).
\end{equation}
known as the \emph{Euler equation}.
\end{proposition}

In particular, every \emph{quadratic operator} in $\g$ corresponds to the geodesic flow of a \emph{right-invariant symmetric linear connection} on $G$. Notice however that this connection is not necessarily Riemannian --- there may not exist a Riemannian metric which is \emph{preserved} by this connection. We will therefore call such an equation as a \emph{non-metric Euler equation}.

\begin{remark}
In the \emph{metric case}, that is when $\nabla$ is the \emph{Levi-Civita connection} of a right invariant Riemannian metric on $G$, the bilinear operator $B$ is related to the metric by the following formula
\begin{equation*}
    B(u,v) = \frac{1}{2}\Big[ (\mathrm{ad}_{u})^{*}(v) + (\mathrm{ad}_{v})^{*}(u)\Big]
\end{equation*}
where $u,v \in \g$ and $(\mathrm{ad}_{u})^{*}$ is the adjoint (relatively to the given Riemannian metric) of the natural action of the Lie algebra on itself given by
\begin{equation*}
    \mathrm{ad}_{u} : v \mapsto [u, v].
\end{equation*}
\end{remark}


\section{The Fr\'{e}chet Lie group $\DiffS$}
\label{sec:diffeomorphisms_group}

In this paper, we are interested in a (non metric) Euler equation on the infinite dimensional Lie group $\DiffS$ of smooth, orientation preserving diffeomorphisms of the unit circle $\Circle$.

Since the tangent bundle
\begin{equation*}
    T\Circle \simeq \Circle \times \Real
\end{equation*}
is trivial, $\VectS$, the space of smooth vector fields on $\Circle$, can be identified with $\CS$, the space of real smooth functions on $\Circle$.

Given a diffeomorphism $\varphi\in\DiffS$, we define $\varphi_{x}\in\CS$ by the following construction. Let $p:\Real \to \Circle$, $x \mapsto \exp (2i\pi x)$ be the universal cover of the circle. A lift of $\varphi$ is a smooth map $f:\Real \to \Real$ such that
\begin{equation*}
    \varphi(\exp (2i\pi x)) = \exp (2i\pi f(x)).
\end{equation*}
Since $\varphi$ preserves the orientation we have $f(x+1) = f(x) + 1$. A lift is not unique but two lifts of $\varphi$ differ by an integer. For any lift $f$ of $\varphi$, its derivative $f^{\prime}$ is a smooth periodic map on $\Circle$ and we define $\varphi_{x}$ as
\begin{equation*}
    \varphi_{x} : = f^{\prime} ,
\end{equation*}
which is independent of the choice of a particular lift of $\varphi$.

\medskip

The space $\CS$ is a \emph{Fr\'{e}chet space}\footnote{A topological vector space $E$ has a canonical \emph{uniform structure}. When this structure is \emph{complete} and when the topology of $E$ may be given by a countable family of \emph{semi-norms}, we say that $E$ is a \emph{Fr\'{e}chet vector space}. A Fr\'{e}chet space is a Banach space if and only if it is locally bounded, which is not the case of $\CS$.}. A family of semi-norms which defines the topology of $\CS$ is just given by the $C^{n}$-norms. A sequence $u_{k}$ converges to $u$ in $\CS$ if and only if
\begin{equation*}
    \norm{u_{k} - u}_{\C{n}} \to 0
\end{equation*}
as $k \to \infty$ for all $n \ge 0$.

\begin{remark}
In a Fr\'{e}chet space, only \emph{directional} or \emph{G\^{a}teaux} derivative is defined
\begin{equation*}
    Df(x)u = \lim_{\varepsilon \to 0} \frac{1}{\varepsilon}(f(x+\varepsilon u) - f(x)).
\end{equation*}
A map $f: X \to Y$ between Fr\'{e}chet spaces $X,Y$ is \emph{continuously differentiable} ($C^{1}$) on $U\subset X$ if the directional derivative $Df(x)u$ exists for all $x$ in $U$ and
all $u$ in $X$, and the map $(x,u) \mapsto Df(x)u$ is continuous. Note that even in the case where $X$ and $Y$ are Banach spaces this definition of continuous differentiability is weaker than the usual one \cite{Ham82}. Higher derivatives and $C^{n}$ classes in Fr\'{e}chet spaces are defined inductively.
\end{remark}

The group $\DiffS$ is naturally equipped with a \emph{Fr\'{e}chet manifold}\footnote{\emph{Fr\'{e}chet manifolds} are defined as sets which can be covered by charts taking values in a given Fr\'{e}chet space and such that the change of charts are smooth.} structure modeled over the \emph{Fr\'{e}chet vector space} $\CS$. A smooth atlas with only two charts may be constructed as follows (see \cite{GR05}). Given $\varphi \in \DiffS$, it is always possible to find a lift $f:\Real \to \Real$ of $\varphi$ such that either
\begin{equation}\label{eq:lift_1}
    -1/2 < f(0) < 1/2,
\end{equation}
or
\begin{equation}\label{eq:lift_2}
    0 < f(0) < 1 ,
\end{equation}
these conditions being not exclusive. Let $V_{1}$ and $V_{2}$ be the subsets of $\varphi \in \DiffS$ for which there is a lift satisfying respectively \eqref{eq:lift_1} or \eqref{eq:lift_2}. Notice that $V_{1}$ and $V_{2}$ are open and that $V_{1} \cup V_{2} = \DiffS$. For any lift $f$ of $\varphi \in \DiffS$, the function
\begin{equation*}
    u = f - id
\end{equation*}
is $1$-periodic and hence lies in $\CS$ and moreover $u^{\prime}(x) > -1$. Let
\begin{equation*}
    U_{1} = \set{ u \in \CS : \: -1/2 < u(0) <1/2 \: \text{and} \: u^{\prime} > -1 }
\end{equation*}
and
\begin{equation*}
    U_{2} = \set{ u \in \CS : \: 0 < u(0) < 1 \: \text{and} \: u^{\prime} > -1 }.
\end{equation*}
For $j=1,2$, the maps
\begin{equation*}
    \Phi_{j} : U_{j} \to V_{j}, \quad u \mapsto f = \id + u, \qquad j=1,2
\end{equation*}
define charts of $\DiffS$ with values in $\CS$. The change of charts corresponds to a change of lift and is just a translation in $\CS$ by $\pm 1$.

Since the composition and the inverse are smooth maps we say that $\DiffS$ is a \emph{Fr\'{e}chet-Lie group} \cite{Ham82}. In particular, $\DiffS$ is itself \emph{parallelizable}
\begin{equation*}
    T\DiffS \simeq \DiffS \times \CS .
\end{equation*}

The tangent space at the unit element $T_{\id}\DiffS$ --- the \emph{Lie algebra} of $\DiffS$ --- is defined as follows. Let $t \mapsto \varphi(t)$ be a smooth path in $\DiffS$ with $\varphi(0) = \id$. Then
\begin{equation*}
    \varphi_{t}(0,x) \in T_{x}\Circle, \qquad \forall x \in T_{x}\Circle
\end{equation*}
and therefore $\varphi_{t}(0,\cdot)$ is a vector field on $\Circle$. In other words, the Lie algebra of $\DiffS$ corresponds to $\VectS$. The \emph{Lie bracket} on $\VectS \simeq \CS$ is given by\footnote{Notice that this bracket differs from the usual bracket of vector fields by a sign.}
\begin{equation*}
    [u,v] = u_{x}v - uv_{x} .
\end{equation*}

\begin{remark}
For $n\ge 1$, we define the group $\Diff{n}$ of orientation-preserving diffeomorphisms of class $C^{n}$. This group is equipped with a smooth \emph{Banach manifold} structure modeled on the Banach vector space $\C{n}$. However, $\Diff{n}$ is only a \emph{topological group} and not a \emph{Banach Lie group} since the composition and the inverse are continuous but not differentiable \cite{EM70}.
\end{remark}


\section{The b-equation as a non-metric Euler equation on $\DiffS$}
\label{sec:b-equation}

We can recast equation~\eqref{eq:b-equation} as
\begin{equation}\label{eq:b-equation_euler}
    u_{t} = -A^{-1}\left[ u(Au)_{x} + b(Au)u_{x}  \right],
\end{equation}
where
\begin{equation*}
    A = I - \frac{d^{2}}{dx^{2}}
\end{equation*}
is an invertible, linear, differential operator on $\CS$. As a quadratic evolution equation on the Lie algebra $\VectS$, each member of the $b$-family corresponds to a the \emph{Euler equation} of a \emph{right invariant symmetric linear connection} on $\DiffS$.

On a Fr\'{e}chet manifold, only covariant derivatives along curves are meaningful. In the present case, the covariant derivative of the vector field
\begin{equation*}
    \xi(t) = \big(\varphi(t),w(t)\big)
\end{equation*}
along the curve $\varphi(t) \in \DiffS$ is defined as
\begin{equation*}
    \frac{D\xi}{Dt}(t) = \left( \varphi(t), w_{t} + \frac{1}{2}[u(t),w(t)] + B\big(u(t),w(t)\big) \right)
\end{equation*}
where $u(t) = \varphi_{t}\circ \varphi^{-1}$ and
\begin{equation*}
    B(u,w) = \frac{1}{2} A^{-1}\left[ u(Aw)_{x} + w(Au)_{x} + b(Au)w_{x} + b(Aw)u_{x}  \right].
\end{equation*}

This connection is Riemannian \emph{only} for $b=2$ (see \cite{Kol09}), that is for (CH). In that case it corresponds to the Levi-Civita connection of the right invariant metric on $\DiffS$ generated by the $H^{1}$ inner product on $\CS$
\begin{equation*}
    <u,v>_{H^{1}} = \int_{\Circle} \left( uv + u_{x}v_{x} \right)\, dx  = \int_{\Circle} u A(v)\, dx
\end{equation*}
where $u,v \in \CS$ and the operator
\begin{equation*}
    B(u,u) = A^{-1}\left[ (Au_{x})u + 2(Au)u_{x} \right]
\end{equation*}
is defined as the $H^{1}$-adjoint of the natural action $\mathrm{ad}_{u}v = u_{x}v - uv_{x}$. That is
\begin{equation*}
    \int_{\Circle} (u_{x}v - uv_{x}) A(w) \, dx = \int_{\Circle} B(w,u) A(v),
\end{equation*}
where $u,v,w \in \CS$.

Euler equation~\eqref{eq:b-equation_euler} is not an \emph{ordinary differential equation} (ODE) on $\C{n}$ because of the term $(Au)_{x}$ which is not regularized by the operator $A^{-1}$ of order $-2$ (the right hand side of \eqref{eq:b-equation_euler} does not belong to $\C{n}$). As we shall show now, the introduction of \emph{Lagrangian coordinates} allows to rewrite equation~\eqref{eq:b-equation_euler} as a second order vector field on the Banach manifold $\Diff{n}$ for $n \ge 2$.

\medskip

Let $J$ be an open interval in $\Real$ and $u$ be a time dependent vector field on $J \times \Circle$ of class $C^{n}$ ($n\ge 1$). Then, the flow $\varphi$ of $u$ is defined on $J \times \Circle$ and is of class $C^{n}$ as well as $\varphi_{t} = u \circ \varphi$. Conversely, let $(\varphi,v)$ be defined, of class $C^{n}$ on $J \times \Circle$ and such that $\varphi(t)$ is a diffeomorphism for all $t\in J$. Then $\varphi^{-1}$ is of class $C^{n}$ on $J \times \Circle$ and similarly for $u = v \circ \varphi^{-1}$.

\begin{proposition}
The $C^{3}$ time dependent vector field $u$ is a solution of \eqref{eq:b-equation_euler} if and only if $(\varphi,v)$ is solution of
\begin{equation}\label{eq:b-equation_geodesic}
    \left\{%
\begin{array}{ll}
    \varphi_{t} & = v, \\
    v_{t} & = - P_{\varphi}(v), \\
\end{array}%
\right.
\end{equation}
where
\begin{equation*}
    P_{\varphi} = R_{\varphi} \circ P \circ R_{\varphi^{-1}}
\end{equation*}
and
\begin{equation*}
 P(u) = A^{-1}\left[ 3u_{x}u_{xx} + b(Au)u_{x}  \right]
\end{equation*}
\end{proposition}

\begin{proof}
Let $u$ be a time dependent vector field on $J \times \Circle$ and $\varphi$ its flow. Let $v := \varphi_{t} = u \circ \varphi$, then $v_{t} = (u_{t} + uu_{x}) \circ \varphi$. Therefore, $u$ is a solution of equation~\eqref{eq:b-equation_euler} if and only if
\begin{align*}
    u_{t} + uu_{x} & = -A^{-1}\left[ u(Au)_{x} - A(uu_{x}) + b(Au)u_{x}  \right] \\
		& = -A^{-1}\left[ 3u_{x}u_{xx} + b(Au)u_{x}  \right].
\end{align*}
That is $u$ is a solution of equation~\eqref{eq:b-equation_euler} if and only if $(\varphi,v)$ is a solution of equation~\eqref{eq:b-equation_geodesic}.
\end{proof}

Since $P$ is a pseudo-differential operator of order $0$, the mapping
\begin{equation*}
    F(\varphi,v) = \Big( v, P_{\varphi}(v) \Big)
\end{equation*}
sends $\Diff{n} \times \C{n}$ into $\C{n} \times \C{n}$. Notice however, that since the trivialization $T\Diff{n} \simeq \Diff{n} \times \C{n}$ is only \emph{topological} and not \emph{smooth}, equation~\eqref{eq:b-equation_geodesic} is meaningful only \emph{locally} (strictly speaking, $\varphi$ should be replaced by its lift $f$ in this expression). In any case, $F$ is a well-defined \emph{second order vector field} on the Banach manifold $\Diff{n}$.

\begin{remark}\label{rem:equivariance}
By its very definition, $F$ is \emph{equivariant} by the right action of $\Diff{n}$ on $\Diff{n} \times \C{n}$. That is
\begin{equation}\label{eq:equivariance}
    F(\varphi \circ \psi, v \circ \psi) = F(\varphi,v) \circ \psi
\end{equation}
for all $(\varphi,v) \in \Diff{n} \times \C{n}$ and $\psi \in \Diff{n}$.
\end{remark}

\begin{theorem}\label{thm:smoothness}
The vector field
\begin{equation*}
    F : \Diff{n} \times \C{n} \to \C{n} \times \C{n}
\end{equation*}
is smooth for all $n\ge 2$.
\end{theorem}

We cannot conclude directly from the smoothness of $P$ that $F$ is smooth because \emph{neither} the composition nor the inversion are smooth maps on $\Diff{n}$. The Banach manifold $\Diff{n}$ is only a topological group. The main argument in the proof of Theorem~\ref{thm:smoothness} is the following lemma.

\begin{lemma}\label{lem:rationality}
Let $P$ be a polynomial differential operator of order $r$ on $\CS$, with \emph{constant coefficients}. Then
\begin{equation*}
    P_{\varphi} = R_{\varphi} \circ P \circ R_{\varphi^{-1}}
\end{equation*}
is a polynomial differential operator of order $r$ whose coefficients are rational functions of $\varphi_{x}, \dotsc \varphi_{x}^{(r)}$.
\end{lemma}

\begin{proof}
It suffices to prove the lemma for a monomial
\begin{equation*}
 P(u) = u^{p_{0}}(u^{\prime})^{p_{1}}\dotsb (u^{(r)})^{p_{r}}.
\end{equation*}
We have
\begin{equation*}
 P_{\varphi}(u) = u^{p_{0}} \left[ (u \circ \varphi^{-1})^{\prime} \circ \varphi \right]^{p_{1}} \dotsb \left[ (u \circ \varphi^{-1})^{(r)}\circ \varphi \right]^{p_{r}}.
\end{equation*}
Let
\begin{equation*}
 a_{k} = (u \circ \varphi^{-1})^{(k)} \circ \varphi , \qquad k = 1, 2, \dotsc
\end{equation*}
Since
\begin{equation*}
 a_{1} = \dfrac{u_{x}}{\varphi_{x}}, \quad \text{and} \quad a_{k+1} = \dfrac{1}{\varphi_{x}}(a_{k})^{\prime}
\end{equation*}
the proof is achieved by a recursive argument.
\end{proof}

\begin{proof}[Proof of Theorem~\ref{thm:smoothness}]
To prove the smoothness of $F$, we will show that for each $n\geq 2$, the operator
\begin{equation*}
    \tilde{P}(\varphi,v) = \Big( \varphi , R_{\varphi} \circ P \circ R_{\varphi^{-1}} (v) \Big)
\end{equation*}
is a smooth map from $\Diff{n} \times \C{n}$ to itself. To do so, we write $\tilde{P}$ as the composition $\tilde{P} = \tilde{A}^{-1} \circ \tilde{Q}$, where
\begin{align*}
    \tilde{A} (\varphi,v) = \Big( \varphi,\,R_{\varphi} \circ A \circ R_{\varphi^{-1}}(v) \Big)
    \intertext{and}
    \tilde{Q} (\varphi,v) = \Big( \varphi,\, R_{\varphi} \circ Q \circ R_{\varphi^{-1}}(v) \Big)
\end{align*}
where $Q$ is the quadratic differential operator defined by
\begin{equation*}
    Q(u) :=  3 u_xu_{xx}+ b (Au)u_x .
\end{equation*}

By virtue of Lemma~\ref{lem:rationality}, both $\tilde{A}$ and $\tilde{Q}$ are smooth maps from $\Diff{n} \times \C{n}$ to $\Diff{n} \times \C{n-2}$, since the spaces $\C{k}$ are Banach algebras.

To show that $\tilde{A}^{-1}: \Diff{n} \times \C{n-2} \to \Diff{n} \times \C{n}$ is smooth, we compute the derivative of $\tilde{A}$ at an arbitrary point
$(\varphi,v)$, obtaining
\begin{equation*}
    D\tilde{A}(\varphi,v) =
    \left( \begin{array}{cc}
        Id & 0 \\
        \ast & R_{\varphi} \circ A \circ R_{\varphi^{-1}}
    \end{array} \right)
\end{equation*}
which is clearly a bounded linear operator from $\C{n} \times \C{n}$ to $\C{n} \times \C{n-2}$ by virtue of Lemma~\ref{lem:rationality}. It is moreover a topological linear isomorphism since it is bijective (\emph{open mapping theorem}). The application of the \emph{inverse mapping theorem} \cite{Lan99} in Banach spaces achieves the proof.
\end{proof}


\section{Short time existence of the geodesics on $\DiffS$}
\label{sec:short_time_existence}

From the smoothness of $F$ and the \emph{Cauchy-Lipschitz theorem} in Banach spaces (see \cite{Lan99} for instance) we deduce existence of integral curves for the vector field $F$ on $\Diff{n} \times \C{n}$.

\begin{proposition}\label{prop:Cauchy_theorem}
Let $n \ge 2$. Then there exists a open interval $J_{n}$ centered at $0$ and an open ball $B_{n}(0,\delta_{n})$ in $\C{n}$ such that for each $\uo \in B(0,\delta_{n})$ there exists a unique solution $(\varphi , v) \in C^{\infty}\big(J_{n},\Diff{n} \times \C{n}\big)$ of \eqref{eq:b-equation_geodesic} such that $\varphi(0) = \id$ and $v(0) = \uo$. Moreover, the flow $(\varphi , v)$ depends smoothly on $(t,\uo)$.
\end{proposition}

Classical results like the \emph{Cauchy-Lipschitz theorem} or the \emph{local inverse theorem} for Banach spaces are no longer valid in general for Fr\'{e}chet spaces. To establish a short time existence result in the $C^{\infty}$ category, it is useful to introduce the notion of \emph{Banach approximation} of Fr\'{e}chet spaces.

\begin{definition}
A \emph{Banach approximation} of a Fr\'{e}chet space $X$ is a sequence of Banach spaces $(X_{n},\norm{{\cdot}}_{n})_{n \ge 0}$ such that
\begin{equation*}
    X_{0} \supseteq X_{1} \supseteq X_{2} \supseteq \dotsb \supseteq X \quad \text{and} \quad X = \bigcap_{n = 0}^{\infty} X_{n}
\end{equation*}
where $\set{ \norm{{\cdot}}_{n} }_{n \ge 0}$ is a sequence of norms inducing the
topology on $X$ and
\begin{equation*}
    \norm{x}_0 \le \norm{x}_1 \le \norm{x}_2 \le \dotsb
\end{equation*}
for any $x$ in $X$.
\end{definition}

For Fr\'{e}chet spaces admitting Banach approximations one has the following lemma.

\begin{lemma}\label{lem:Banach_approximation}
Let $X$ and $Y$ be Fr\'{e}chet spaces with Banach approximations $(X_{n})_{n \ge 0}$, and $(Y_{n})_{n\ge 0}$ respectively. Let $\Phi_{0} : U_{0} \to V_{0}$ be a smooth map between the open subsets $U_{0} \subset X_0$ and $V_{0} \subset Y_{0}$. Let $U = U_{0} \cap X$, $V = V_{0} \cap Y$ and for each $n \ge 0$
\begin{equation*}
    U_{n} = U_{0} \cap X_{n}, \quad V_{n} = V_{0} \cap Y_{n}.
\end{equation*}
Assume that, for each $n \ge 0$, the following properties are satisfied:
\begin{enumerate}
    \item $\Phi_{0}(U_{n}) \subset Y_{n}$,
    \item the restriction $\Phi_{n} := \Phi_{0}\big\arrowvert_{U_{n}} : U_{n} \to Y_{n}$ is a smooth map.
\end{enumerate}
Then one has $\Phi_{0}(U)\subset V$ and the map $\Phi := \Phi_{0} \big\arrowvert_{U} : U \to V$ is smooth.
\end{lemma}

\begin{proof}
By property (1) we have $\Phi_{0}(U)\subset V$ and $\Phi = \Phi_{0} \big\arrowvert_{U} : U \to V$ is clearly continuous. Property (2) ensures that the \emph{G\^{a}teaux} derivative of $\Phi$ is well-defined for every $u\in U$ and for any direction $w$ in $X$ and that $(u,w) \mapsto D\Phi(u)w$ is continuous. Inductively, property (2) shows that $\Phi$ is of class $C^{n}$ for all $n$, and therefore is a smooth map.
\end{proof}

A remarkable property of equation~\eqref{eq:b-equation} whose interpretation remains mysterious (unless for $b=2$ where it corresponds to the conservation of the momentum for \emph{metric Euler} equations) is the following conservation law.

\begin{proposition}\label{prop:conservation_law}
Let $u$ be a solution of \eqref{eq:b-equation} of class $C^{3}$ on $[0,T]\times \Circle$ and $\varphi$ the flow of the time-dependent vector field $u$. Then we have
\begin{equation*}
 (m(t) \circ \varphi(t))\varphi_{x}^{b}(t) = m(0)
\end{equation*}
for all $t \in [0,T]$.
\end{proposition}

\begin{proof}
We have
\begin{equation*}
    \frac{d}{dt}\Big[(m \circ \varphi)\varphi_{x}^{b}\Big] = \varphi_{x}^{b} \Big\{ m_{t} \circ \varphi + (m_{x} \circ \varphi) (u \circ \varphi) + b (m \circ \varphi)(u_{x} \circ \varphi) \Big\}
\end{equation*}
since $\varphi_{t} = u \circ \varphi$ and $\varphi_{xt} = (u_{x} \circ \varphi) \varphi_{x}$. Hence
\begin{equation*}
    \frac{d}{dt}\Big[(m \circ \varphi)\varphi_{x}^{b}\Big] = 0
\end{equation*}
if $m_{t} = - (m_{x}u + bmu_{x})$.
\end{proof}

\begin{lemma}\label{lem:regularity}
Let $(\varphi(t),v(t))$ be a solution of \eqref{eq:b-equation_geodesic} in $\Diff{3}\times \C{3}$ defined on $[0,T]$ with initial data $(\id,\uo)$. Then we have for all $t \in[0,T]$
\begin{equation}\label{eq:regularity_1}
    \varphi_{xx}(t) = \varphi_{x}(t)\left[ \int_{0}^{t} v(s) \varphi_{x}(s) \, ds - m_{0} \int_{0}^{t} \varphi_{x}(s)^{1-b} \, ds \right]
\end{equation}
and
\begin{multline}\label{eq:regularity_2}
    v_{xx}(t) = v_{x}(t)\left[ \int_{0}^{t} v(s) \varphi_{x}(s) \, ds - m_{0} \int_{0}^{t} \varphi_{x}(s)^{1-b} \, ds \right] \\
    + \varphi_{x}(t)\left[ v(t) \varphi_{x}(t) - m_{0} \varphi_{x}(t)^{1-b} \right].
\end{multline}
\end{lemma}

\begin{proof}
We have
\begin{equation*}
    \frac{d}{dt}\left(\frac{\varphi_{xx}}{\varphi_{x}}\right) = (u_{xx} \circ \varphi) \varphi_{x}
\end{equation*}
since
\begin{equation*}
    \varphi_{xt} = (u_{x} \circ \varphi) \varphi_{x} \quad \text{and} \quad \varphi_{xxt} = (u_{xx} \circ \varphi) \varphi_{x}^{2} + (u_{x} \circ \varphi)\varphi_{xx}.
\end{equation*}
Using Proposition~\ref{prop:conservation_law}, we get
\begin{equation*}
    u_{xx} \circ \varphi = u \circ \varphi - m_{0}\varphi_{x}^{-b}
\end{equation*}
and hence
\begin{equation*}
    \frac{d}{dt}\left(\frac{\varphi_{xx}}{\varphi_{x}}\right) = (u \circ \varphi)\varphi_{x} - m_{0}\varphi_{x}^{1-b}.
\end{equation*}
Integrating this last relation on the interval $[0,t]$ leads to equation~\eqref{eq:regularity_1} and taking the time derivative of \eqref{eq:regularity_1} leads to equation~\eqref{eq:regularity_2}.
\end{proof}

The meaning of Lemma~\ref{lem:regularity} is that there can be ``neither loss, nor gain'' in spatial regularity: the solution at time $t>0$ has \emph{exactly the same regularity} as it has at time $t=0$.

\begin{corollary}\label{cor:regularity}
Let $n \ge 3$ and let $(\varphi,v)$ be a solution of
\eqref{eq:b-equation_geodesic} in $\Diff{3} \times \C{3}$ with initial data $(\id,\uo)$.
\begin{enumerate}
 \item If $\uo \in \C{n}$ then $(\varphi (t) , v(t)) \in  \Diff{n} \times \C{n}$ for all $t \in [0,T]$.
 \item If there exists $t \in (0,T]$ such that $\varphi (t) \in \Diff{n}$ or $v(t) \in \C{n}$ then $\uo \in \C{n}$.
\end{enumerate}
\end{corollary}

\begin{proof}
1) Let $n \ge 3$. From Lemma~\ref{lem:regularity} we get that if $(\varphi (t) , v(t)) \in  \Diff{n} \times \C{n}$ for all $t \in [0,T]$ and $\uo \in \C{n+1}$ then $(\varphi (t) , v(t)) \in  \Diff{n+1} \times \C{n+1}$ for all $t \in [0,T]$. A recursive argument completes the proof of the first assertion of the corollary.

2) Let $n \ge 3$ and suppose that $\uo\in \C{n}$. Then, as we just verified, $(\varphi (t) , v(t)) \in  \Diff{n} \times \C{n}$ for all $t \in [0,T]$. If moreover there exists $t \in (0,T]$ such that $\varphi (t) \in \Diff{n+1}$ or $v(t) \in \C{n+1}$, then since $\varphi_x$ is strictly positive, we get from Lemma~\ref{lem:regularity} that $\uo\in \C{n+1}$. Again, a recursive argument completes the proof of the second assertion of the corollary.
\end{proof}

\begin{remark}
Proposition~\ref{prop:conservation_law}, Lemma~\ref{lem:regularity} and Corollary~\ref{cor:regularity} are of course true if we replace the time interval $[0,T]$ by $[-T,0]$.
\end{remark}

We are know able to state our main theorem.

\begin{theorem}\label{thm:short_time_existence}
There exists an open interval $J$ centered at $0$ and $\delta >0$ such that for each $\uo \in \CS$ with $\norm{\uo}_{\C{3}}< \delta$, there exists a unique solution $(\varphi , v) \in C^{\infty}\big(J,\DiffS \times \CS\big)$ of \eqref{eq:b-equation_geodesic} such that $\varphi(0) = \id$ and $v(0) = \uo$. Moreover, the flow $(\varphi , v)$ depends smoothly (in the smooth category) on $(t,\uo)$ in $J \times \CS$.
\end{theorem}

\begin{proof}
By application of Proposition~\ref{prop:Cauchy_theorem} for $n=3$, we obtain the existence of an open interval  $J$ centered at $0$ and an open ball $U_{3} = B_{3}(0,\delta)$ in $\C{3}$ such that for each $\uo \in U_3$ there exists a unique solution $(\varphi , v) \in C^{\infty}\big(J,\Diff{3} \times \C{3}\big)$ of \eqref{eq:b-equation_geodesic} with initial data $(\id,\uo)$ and such that the flow
\begin{equation*}
    \Phi_{3}: J \times U_{3} \to \Diff{3} \times \C{3}
\end{equation*}
is smooth. Let $U_{n} = U_{3} \cap \C{n}$ and $U_{\infty} = U_{3} \cap \CS$.

By Corollary~\ref{cor:regularity}, we have
\begin{equation*}
    \Phi_{3}(J \times U_{n}) \subset \Diff{n} \times \C{n},
\end{equation*}
for each $n\ge 3$ and by Cauchy-Lipschitz theorem, the map
\begin{equation*}
    \Phi_{n} := \Phi_{3}\big\arrowvert_{J \times U_{n}} : J \times U_{n} \to \Diff{n} \times \C{n}
\end{equation*}
is smooth. Therefore, applying Lemma~\ref{lem:Banach_approximation}, we get first that
\begin{equation*}
    \Phi_{3}\left( J \times U_{\infty} \right) \subset \DiffS \times \CS
\end{equation*}
which shows the short time existence in the smooth category and then that the map
\begin{equation*}
    \Phi_{\infty} := \Phi_{3} \big\arrowvert_{J \times U_{\infty}} : J \times U_{\infty} \to \DiffS \times \CS
\end{equation*}
is smooth, which shows the smooth dependence upon initial condition in this category. This achieves the proof.
\end{proof}

In the smooth category, the map
\begin{equation*}
    \DiffS \times \CS \to \CS, \qquad (\varphi,v) \mapsto v \circ \varphi^{-1}
\end{equation*}
is smooth and we immediately get the following corollary.

\begin{corollary}\label{cor:short_time_existence_euler}
There exists an open interval $J$ centered at $0$ and $\delta >0$ such that for each $\uo \in \CS$ with $\norm{\uo}_{\C{3}}< \delta$, there exists a unique solution $u \in C^{\infty}\big(J,\CS\big)$ of \eqref{eq:b-equation} such that $u(0) = \uo$. Moreover, the solution $u$ depends smoothly (in the smooth category) on $(t,\uo)$ in $J \times \CS$.
\end{corollary}


\section{The exponential map}
\label{sec:exponential_map}

The geodesic flow of a symmetric linear connection on a Banach manifold $M$ (also called a \emph{spray} in \cite{Lan99}) satisfies the following remarkable property
\begin{equation*}
 \varphi(t, x_{0}, sv_{0}) = \varphi(st, x_{0}, v_{0}),
\end{equation*}
like in the Riemannian case. This is in fact a consequence of the quadratic nature of the geodesic equation. Therefore, the \emph{exponential map} $\e_{x_{0}}$, defined as the time one of the flow is well defined in a neighbourhood of $0$ in $T_{x_{0}}M$ for each point $x_{0}$. On a \emph{Banach manifold}, it can be shown moreover (see \cite{Lan99} for instance) that the differential of $\e_{x_{0}}$ at $0$ is equal to $\id$ and therefore that $\e_{x_{0}}$ is a local diffeomorphism from a neighbourhood of $0$ in $T_{x_{0}}M$ to a neighbourhood of $x_{0}$ in $M$. This privileged chart, called the \emph{normal chart} plays a very special role in classical differential geometry, especially to establish \emph{convexity results}.

On a \emph{Fr\'{e}chet manifold} and in particular on $\DiffS$, the existence of this privileged chart is far from being granted automatically. One may find useful to recall on this occasion that the \emph{group exponential} of $\DiffS$ is not a local diffeomorphism \cite{Mil84}. Moreover, the Riemannian exponential map for the $L^{2}$ metric (Burgers equation) on $\DiffS$ is not a local $C^{1}$-diffeomorphism near the origin \cite{CK02}.

However, it has been established in \cite{CK02}, that for the Camassa-Holm equation -- which corresponds to the Euler equation of the $H^{1}$ metric on $\DiffS$ -- and more generally for $H^k$ metrics ($k\ge 1$) (see \cite{CK03}), the Riemannian exponential map was in fact a smooth local diffeomorphism. We extend this result here for the general $b$-equation (a non metric Euler equation).

\begin{theorem}\label{thm:exponential_map}
The exponential map $\e $ at the unit element $\id$ for the $b$-equation on $\DiffS$ is a smooth local diffeomorphism from a neighborhood of zero in $\VectS$ to a neighborhood of $\id$ on $\DiffS$.
\end{theorem}

The proof of this theorem relies mainly on a linearized version of Lemma~\ref{lem:regularity}. Let $n \ge 3$ and $(\varphi^{\varepsilon}(t),v^{\varepsilon}(t))$ be the local expression of an integral curve in $T\Diff{n}$ of \eqref{eq:b-equation_geodesic}, defined on $[0,T]$, with initial data $(\id,u + \varepsilon w)$, where $u, w \in \C{n}$. We define
\begin{equation*}
    \psi(t) = \left.\frac{\partial}{\partial \varepsilon}\right|_{\varepsilon = 0} \varphi^{\varepsilon}(t) .
\end{equation*}
Then $\psi(t) = L_{n}(t,u)w$ where $L_{n}(t,u)$ is a bounded linear operator of $\C{n}$.

\begin{lemma}\label{lem:linearized_regularity}
Suppose $u \in \C{n+1}$. Then, we have
\begin{equation*}
    L_{n}(t,u)(\C{n} \backslash \C{n+1}) \subset \C{n} \backslash \C{n + 1})
\end{equation*}
for all $t \in(0,T]$,
\end{lemma}

\begin{proof}
Writing \eqref{eq:regularity_1} for $(\varphi^{\varepsilon}(t),v^{\varepsilon}(t))$ and taking the derivative with respect to $\varepsilon$ at $\varepsilon =0$, we get
\begin{multline*}
    \psi_{xx}(t) = a(t)\psi_{x}(t) + b(t) \psi(t) - c(t) (w - w_{xx}) \\
      + \int_{0}^{t}\alpha(t,s)\psi(s) \,ds + \int_{0}^{t}\beta(t,s)\psi_{x}(s) \,ds
\end{multline*}
where $a(t),b(t),c(t), \alpha(t,s), \beta(t,s)$ are in $\C{n-1}$ and $c(t) >0$ for $t > 0$. Therefore, if
\begin{equation*}
    w \in \C{n}\setminus \C{n+1}
\end{equation*}
then
\begin{equation*}
    \psi(t) = L_{n}(t,u)w \in \C{n}\setminus \C{n+1} , \qquad \forall t \in (0,T],
\end{equation*}
which achieves the proof of the lemma.
\end{proof}

\begin{proof}[Proof of Theorem~\ref{thm:exponential_map}]
First, we can find neighborhoods $U_{3}$ of $0$ in $\C{3}$ and $V_{3}$ of $\id$ in $\Diff{3}$ such that
\begin{equation*}
    \e_{3} : U_{3} \to V_{3}
\end{equation*}
is a smooth diffeomorphism (for the $C^{3}$ norm). For $n \ge 3$, let
\begin{equation*}
    U_{n} = U_{3} \cap \C{n} \quad \text{and} \quad V_{n} = V_{3} \cap \Diff{n}.
\end{equation*}
By Corollary~\ref{cor:regularity}, $ \e_{3}(U_{n})= V_{n}$ and
\begin{equation*}
    \e_{n}: = \e_{3} \big\arrowvert_{U_{n}} : U_{n} \to V_{n}
\end{equation*}
is a bijection. By virtue of the \emph{Cauchy-Lipschitz theorem} in Banach spaces, $\e_{n}$ is smooth (for the $C^{n}$ norm). We are going to show that it is a diffeomorphism. For each $u\in \C{n}$, $D\e_{n}(u)$ is a bounded linear operator of $\C{n}$. Notice that
\begin{equation*}
    D\e_{n}(u) = D\e_{3}(u)\big\arrowvert_{\C{n}}.
\end{equation*}
It is therefore \emph{one-to-one}. We now prove inductively that $D\e_{n}(u)$ is surjective. For $n=3$ this is so by our hypothesis. If it is true for $3 \le j \le n$, then it true also for $n+1$ (with $u\in \C{n+1}$) because of Lemma~\ref{lem:linearized_regularity} and the fact that $D\e_{n}(u) = L_{n}(1,u)$. Therefore, according to the \emph{open mapping theorem}, we get that for every $n\ge 3$ and each $u\in \C{n}$
\begin{equation*}
    D\e_{n}(u) : \C{n} \to \C{n}
\end{equation*}
is a topological linear isomorphism. Applying the \emph{inverse function theorem} for $\e_{n}$, we deduce that $\e_{n}$ is a diffeomorphism from $U_{n}$ to $V_{n}$. Since this is true for all $n \ge 3$, we conclude, using Lemma~\ref{lem:Banach_approximation} that
\begin{equation*}
    \e_{\infty} := \e_{3}\big\arrowvert_{U_{\infty}} : U_{\infty} \to V_{\infty}
\end{equation*}
as well as
\begin{equation*}
    \e_{\infty}^{-1} : V_{\infty} \to U_{\infty}
\end{equation*}
are smooth maps. That is $\e_{\infty}$ is a diffeomorphism (in the smooth category) between $U_{\infty}$ and $V_{\infty}$.
\end{proof}


\section{The non-periodic case}
\label{sec:non-periodic_case}

It would be interesting to extend the whole work done in this paper for the non-periodic case --- that is for the group $\DiffR$ of smooth orientation-preserving diffeomorphisms of the real line. Unfortunately and contrary to $\DiffS$, this \emph{Fr\'{e}chet Lie group} is not \emph{regular} as defined by Milnor in \cite{Mil84}. In particular, not every element of its Lie algebra, $\VectR$ can be integrated into a one parameter subgroup.

In order to extend straightforwardly the present work, one needs therefore to restrict to a subgroup of $\DiffR$ which is a \emph{regular Fr\'{e}chet Lie} group. Moreover, the Fr\'{e}chet space on which the differentiable structure is modeled must admit a \emph{Banach approximation}, as defined in Section~\ref{sec:short_time_existence}. And last but not least, these Banach spaces must be \emph{Banach algebras} (for pointwise multiplication of functions) in order to prove Theorem~\ref{thm:smoothness}.

One candidate was proposed in \cite{Mic06}. It consists of the subgroup of \emph{rapidly decreasing diffeomorphisms}
\begin{equation*}
    \mathrm{Diff}_{\mathcal{S}}(\Real) = \set{\id + f;\; f \in \mathcal{S}(\Real) \; \text{and} \; f^{\prime} > -1}
\end{equation*}
where $\mathcal{S}(\Real)$ is the \emph{Schwartz space} of rapidly decreasing functions. This group is a regular Fr\'{e}chet Lie group (see \cite{Mic06}).

Another and simpler solution has been proposed in \cite{DH09}. It consists of the subgroup of $\DiffR$, defined as
\begin{equation*}
    \mathrm{Diff}_{H^{\infty}}(\Real) = \set{\id + f;\; f \in H^{\infty}(\Real) \; \text{and} \; f^{\prime} > -1}
\end{equation*}
where
\begin{equation*}
    H^{\infty}(\Real) = \bigcap_{n=1}^{+ \infty}H^{n}(\Real)
\end{equation*}
and $H^{n}(\Real)$ are the Sobolev spaces on the line. It has been shown in \cite{DH09} that this group is a regular Fr\'{e}chet Lie group and that the theory extends well, at least in the metric case. Notice that
\begin{equation*}
    \mathrm{Diff}_{\mathcal{S}}(\Real) \subset \mathrm{Diff}_{H^{\infty}}(\Real)
\end{equation*}
and that this inclusion is strict.


\end{document}